\documentclass[12pt,a4paper]{article}
\usepackage[utf8]{inputenc}
\usepackage[english]{babel}
\usepackage{amsmath, amssymb, amsfonts}
\usepackage{amsthm} 
\usepackage{physics}
\usepackage[T1]{fontenc}
\usepackage[outline]{contour}
\usepackage{graphicx}
\usepackage{cite}
\usepackage[colorlinks=true, linkcolor=blue, citecolor=red, urlcolor=blue]{hyperref}
\usepackage{hyperref}
\usepackage{tikz}
\usepackage{geometry}
\usepackage{dsfont}
\usepackage{float}
\usepackage{cleveref}
\usepackage{authblk}
\usepackage{lscape}

\theoremstyle{plain}
\newtheorem{theorem}{Theorem}[section]

\newtheorem{proposition}[theorem]{Proposition}
\theoremstyle{definition}
\newtheorem{definition}[theorem]{Definition}
\newtheorem{example}[theorem]{Example}
\theoremstyle{remark}
\newtheorem{remark}[theorem]{Remark}
\title{\textbf{A New $q$-Heisenberg Algebra}}
\author{Julio Cesar Jaramillo Quiceno\\Departamento de Física Universidad Nacional de Colombia\thanks{jcjaramilloq@unal.edu.co}}
\date{}
\begin{document}
\maketitle
\begin{abstract}
This work introduces a novel $q$-$\hbar$ deformation of the Heisenberg algebra, designed to unify and extend several existing $q$-deformed formulations. Starting from the canonical Heisenberg algebra defined by the commutation relation $[\hat{x}, \hat{p}] = i\hbar$ on a Hilbert space \cite{Zettili2009}, we survey a variety of $q$-deformed structures previously proposed by Wess \cite{Wess2000}, Schmüdgen \cite{Schmudgen1999}, Wess--Schwenk \cite{Wess-Schwenk1992}, Gaddis \cite{Jasson-Gaddis2016}, and others. These frameworks involve position, momentum, and auxiliary operators that satisfy nontrivial commutation rules and algebraic relations incorporating deformation parameters. Our new $q$-$\hbar$ Heisenberg algebra $\mathcal{H}_q$ is generated by elements $\hat{x}_\alpha$, $\hat{y}_\lambda$, and $\hat{p}_\beta$ with $\alpha, \lambda, \beta \in \{1,2,3\}$, and is defined through generalized commutation relations parameterized by real constants $n, m, l$ and three dynamical functions $\Psi(q)$, $\Phi(q)$, and $\Pi(q)$ depending on the deformation parameter $q$ and the generators. By selecting appropriate values for these parameters and functions, our framework recovers several well-known algebras as special cases, including the classical Heisenberg algebra for $q = 1$ and $\Psi = 1$, $\Phi = \Pi = 0$, and various $q$-deformed algebras for $q \neq 1$. The algebraic consistency of these generalizations is demonstrated through a series of explicit examples, and the resulting structures are shown to align with quantum planes \cite{Yuri-Manin2010} and enveloping algebras associated with Lie algebra homomorphisms \cite{Reyes2014a}. This construction offers a flexible and unified formalism for studying quantum deformations, with potential applications in quantum mechanics, noncommutative geometry, and quantum group theory.
\end{abstract}
\section{Heisenberg algebras}\label{Heis-alg-sect1}

In this section, we present the definitions and preliminaries on some deformations of the Heisenberg algebra, focusing on the families of Heisenberg algebras of interest in this work.

The usual quantum-mechanical operators in this space will be denoted with a hat. Thus, the element \(\hat{x}\) of the algebra will be identified with the observable for the \textbf{ position} in space, and the element \(\hat{p}\) with the observable conjugate conjugate canonical, typically called the \textbf{ momentum}. The space where this algebra is generated is called the \textit{\textbf{Hilbert space}}, $\mathcal{H}$.

\begin{definition}\cite{Zettili2009} A Hilbert space \( \mathcal{H} \) consists of a set of vectors \( \psi, \phi \) and a set of scalars \( \alpha, \beta \) which satisfy the following four properties 
\begin{enumerate}
    \item[(i)]  \( \mathcal{H} \) is a linear space.
    \item[(ii)]  \( \mathcal{H} \) has a \textit{defined scalar product that is strictly positive}. The scalar product of an element \( \psi \) with another element \( \phi \) is in general a complex number, denoted by \( (\psi, \phi) \), where 
    \[
    (\psi, \phi) = \psi^* \phi \quad \text{(complex number)}.
    \]
    \item[(iii)]  \( \mathcal{H} \) is separable: \( \psi_n \in \mathcal{H} \) (\( n = 1,2,3, \dots \)) such that for every \( \psi \in \mathcal{H} \) and \( \epsilon > 0 \), there exists at least one \( \psi_n \) in the sequence for which 
    \[
    \|\psi - \psi_n\| < \epsilon.
    \]
    \item[(iv)]  \( \mathcal{H} \) is complete: Every Cauchy sequence \( \psi_n \in \mathcal{H} \) converges to an element of \( \mathcal{H} \). That is, for any \( \psi_n \), the relation
    \[
    \lim_{n,m \to \infty} \|\psi_n - \psi_m\| = 0,
    \]
    defines a unique limit \( \psi \in \mathcal{H} \) such that 
    \[
    \lim_{n \to \infty} \|\psi - \psi_n\| = 0.
    \]
\end{enumerate}  
\end{definition}

\begin{definition}\cite{Heisenberg1925, Dirac1925}\label{classic-Heis}
Let \(\hat{x}_{n}\), \(\hat{x}_{m}\) be position operators and \(\hat{p}_{n}\), \(\hat{p}_{m}\) be momentum operators. According to the above, Heisenberg's crucial idea that led to quantum mechanics was to consider the components of these vectors as operators on a Hilbert space \(\mathcal{H}\), satisfying the commutation relations:
\begin{align}
   \label{Heis-1}\hat{x}_{n}\hat{p}_{m} - \hat{p}_{m}\hat{x}_{n} &= i\hbar\delta_{i,j}, \\ 
   \label{Heis-2}\hat{x}_{n}\hat{x}_{m} - \hat{x}_{m}\hat{x}_{n} &= 0, \\ 
   \label{Heis-3}\hat{p}_{n}\hat{p}_{m} - \hat{p}_{m}\hat{p}_{n} &= 0,
\end{align}
for \(m, n = 1, 2, 3\),  and \(\delta_{ij}\) is the Kronecker delta. This set of relations is known as the \textbf{Heisenberg algebra} or the \textit{\textbf{canonical quantization}}.
\end{definition}

\subsection{\texorpdfstring{$q$}{Lg}-deformed Heisenberg Algebra}\label{juwess}

In this section, recall the $q$-deformed Heisenberg algebra following the treatment presented by Wess \cite{Wess2000}, starting from the canonical commutation relations (\ref{Heis-1}), (\ref{Heis-2}) and (\ref{Heis-3}).

\begin{definition}\cite{Wess2000}\label{Wess-def} Let $\hat{x},\hat{p}$, and $\hat{\Lambda}$ be generators. The \textit{\textbf{$q$-deformed Heisenberg algebra}} is defined by the following relations: 
\begin{align} \label{17}
q^{1/2}\hat{x}\hat{p}-q^{-1/2}\hat{p}\hat{x} & = \ i\hat{\Lambda}\hbar ,\\
\label{wess2}\hat{\Lambda} \hat{x} & = \ q^{-1}\hat{x}\hat{\Lambda}, \\
\label{wess3}\hat{\Lambda} \hat{p}  & = \ q\hat{p}\hat{\Lambda}, 
\end{align} 
The elements of this algebra are supposed to be
self adjoint, 
\begin{align}
    \label{wess4}\overline{\hat{p}} & = \ \hat{p},\quad  \overline{\hat{x}} = \hat{x}, \quad \overline{\hat{\Lambda}}=\hat{\Lambda}^{-1},
\end{align}
for $q\in\mathbb{R}$ with $q\neq 0$. 
\end{definition}

If $\hat{x}, \hat{p}$ and $\hat{\Lambda}$ satisfy the relations (\ref{17}), (\ref{wess2}) and (\ref{wess3}), then the $q$-Heisenberg algebra can be written as the free associative algebra
    \begin{equation*}
        \mathcal{H}_{1}^{W} = \frac{\mathbb{C}\left\{ \hat{p}, \hat{x}, \hat{\Lambda} \right\} }{ \langle q^{\frac{1}{2}} \hat{x}\hat{p} - q^{-\frac{1}{2}} \hat{p}\hat{x} - i\hat{\Lambda}\hbar, \ \hat{\Lambda} \hat{p} - q \hat{p} \hat{\Lambda}, \ \hat{\Lambda} \hat{x} - q^{-1} \hat{x} \hat{\Lambda}, \overline{\hat{p}}=\hat{p},\overline{\hat{x}} = \hat{x}, \overline{\hat{\Lambda}}=\hat{\Lambda}^{-1}\mid q\in \mathbb{R}, \ q\neq 0 \rangle}.
    \end{equation*}

\begin{proposition}\label{Rings-wess}
The iterated Ore extensions $R=\mathbb{R}[\hat{\Lambda}][\hat{p};\sigma_{\hat{p}},\delta_{\hat{p}}][\hat{x};\sigma_{\hat{x}},\delta_{\hat{x}}]$ can be written of the form
\begin{align*}
    \sigma_{\hat{x}}(\hat{\Lambda}) & = \ q\hat{\Lambda}, \quad \delta_{\hat{x}}(\hat{\Lambda})=0, \quad \sigma_{\hat{x}}(\hat{p})=q^{-1}\hat{p}, \quad \sigma_{\hat{p}}(\hat{\Lambda})=q^{-1}\hat{\Lambda}, \quad \delta_{\hat{p}}(\hat{\Lambda})=0.
\end{align*}
\end{proposition}
\begin{proof}
   First, we write 
   \begin{equation*}
       \hat{p}\hat{\Lambda}=\sigma_{\hat{p}}(\Lambda)\hat{p}+\delta_{\hat{p}}(\hat{\Lambda}), 
   \end{equation*}
   Compared with (\ref{wess3}) we can say that
   \begin{equation*}
    \sigma_{\hat{p}}(\hat{\Lambda})=q^{-1}\hat{\Lambda}, \quad \delta_{\hat{p}}(\hat{\Lambda})=0.
   \end{equation*}
   Now for (\ref{wess2}) we can proceed in the same form as above
   \begin{equation*}
       \hat{x}\hat{\Lambda}=\sigma_{\hat{x}}(\hat{\Lambda})\hat{x}+\delta_{\hat{x}}(\hat{\Lambda}),
   \end{equation*}
   we can infer that 
   \begin{equation*}
       \sigma_{\hat{x}}(\hat{\Lambda})=q\hat{\Lambda}, \quad \delta_{\hat{x}}(\hat{\Lambda})=0,
   \end{equation*}
   and for (\ref{17})
   \begin{equation*}
       \hat{x}\hat{p} = \sigma_{x}(\hat{p})\hat{x}+\delta_{\hat{x}}(\hat{p}),
   \end{equation*}
   obtaining that
   \begin{equation*}
     \sigma_{\hat{x}}(\hat{p})=q^{-1}\hat{p}, \quad \delta_{\hat{x}}(\hat{p})=i\hbar q^{-1/2}\hbar\hat{\Lambda}.
   \end{equation*}  
\end{proof}

\subsection{\texorpdfstring{$q$}{Lg}-Heisenberg algebra by Schmudgen}\label{rem2}

\begin{definition}\cite{Schmudgen1999}\label{schmudgena}
For a positive real number $q\neq 1$, let $\mathcal{H}^{SCH}$ denote the complex unital algebra with four generators $\hat{p}, \hat{u}, \hat{x}$ and $\hat{u}^{-1}$ subject to the defining relations
\begin{align}
 \label{22}\hat{u}\hat{p} & = \ q\hat{p}\hat{u},\\
 \label{23}\hat{u}\hat{x} & = \ q^{-1}\hat{x}\hat{u},\\
 \label{24}\hat{u}\hat{u}^{-1} & =  \ \hat{u}^{-1}\hat{u}  =  1,\\
 \label{25}\hat{p}\hat{x}-q\hat{x}\hat{p} & = \ i(q^{3/2}-q^{-1/2})\hat{u}\hbar,\\
 \label{26}\hat{x}\hat{p}-q\hat{p}\hat{x} & = \ -i(q^{3/2}-q^{-1/2})\hat{u}^{-1}\hbar.
 \end{align}

\end{definition}
     If $\mathbb{C}\left\{ \hat{u}, \hat{p}, \hat{x}, \hat{u}^{-1}\right\}$ is the free algebra on five operators  $\hat{u}, \hat{p}, \hat{x}$ and $\hat{u}^{-1}$, and $\langle \hat{p}\hat{x} - q\hat{x}\hat{p} - i(q^{3/2}-q^{-1/2})\hat{u}\hbar, \ \hat{x} \hat{p} - q \hat{x} \hat{p}+i(q^{3/2}-q^{-1/2})\hat{u}^{-1}\hbar, \ \hat{u} \hat{u}^{-1} - \hat{u}^{-1} \hat{u}-1, \ \hat{u}\hat{x}-q^{-1}\hat{x}\hat{u}, \ \hat{u}\hat{p}-q\hat{p}\hat{u}, \mid q\in \mathbb{R}, \ q\neq 0 \rangle$ denotes the ideal of $\mathbb{C}\left\{ \hat{u}, \hat{p}, \hat{x}, \hat{x}^{-1}, \hat{u}^{-1}\right\}$ generated by (\ref{22}), (\ref{23}), (\ref{24}), (\ref{25}) and (\ref{26}) we are declaring that
     \begin{tiny}
      \begin{multline*}
          \mathcal{H}^{SCH}=\\
          \frac{\mathbb{C}\left\{ \hat{u}, \hat{p}, \hat{x}, \hat{u}^{-1}\right\}}{\langle \hat{p}\hat{x} - q\hat{x}\hat{p} - i(q^{3/2}-q^{-1/2})\hat{u}\hbar, \ \hat{x} \hat{p} - q \hat{x} \hat{p}+i(q^{3/2}-q^{-1/2})\hat{u}^{-1}\hbar, \ \hat{u} \hat{u}^{-1} - \hat{u}^{-1} \hat{u}-1, \ \hat{u}\hat{x}-q^{-1}\hat{x}\hat{u}, \ \hat{u}\hat{p}-q\hat{p}\hat{u}, \mid q\in \mathbb{R}, \ q\neq 0 \rangle}.
      \end{multline*}
      \end{tiny}
\begin{proposition}[{\cite[Eq. 2]{Schmudgen1999}}]
An equivalent set of relations is obtained if \rm(\ref{25}) and  \rm(\ref{26}) are replaced by
\begin{align}
   \label{27}\hat{p}\hat{x} & = \ iq^{1/2}\hat{u}-iq^{-1/2}\hat{u}^{-1}\hbar,\\
   \label{28}\hat{x}\hat{p} & = \ iq^{-1/2}\hat{u}^{-1}-iq^{1/2}\hat{u}\hbar.
\end{align}
\end{proposition}
\begin{proof}
First, multiply equation (\ref{25}) by $q^{-1}$ in order to obtain
\begin{equation}\label{bonito}
        q^{-1}\hat{p}\hat{x}-\hat{x}\hat{p} =iq^{1/2}\hat{u}\hbar-iq^{-3/2}\hat{u}\hbar-iq^{3/2}\hat{u}^{-1}\hbar+iq^{-1/2}\hat{u}^{-1}\hbar,
\end{equation}

adding (\ref{bonito}) and (\ref{26}) we get that
\begin{align*}
    q^{-1}\hat{p}\hat{x}-q\hat{p}\hat{x} & = \ iq^{1/2}\hat{u}\hbar-iq^{-3/2}\hat{u}-iq^{3/2}\hat{u}^{-1}\hbar+iq^{-1/2}\hat{u}^{-1}\hbar,\\
    \hat{p}\hat{x} & = \ i\left(\frac{q^{1/2}-q^{-3/2}}{q^{-1}-q}\right)\hat{u}\hbar+i\left(\frac{q^{-1/2}-q^{3/2}}{q^{-1}-q}\right)\hat{u}^{-1}\hbar,
\end{align*}

A short calculation going through the steps
\begin{equation*}
  \left(\frac{q^{1/2}-q^{-3/2}}{q^{-1}-q}\right)=-q^{-1/2}, \quad  \left(\frac{q^{-1/2}-q^{3/2}}{q^{-1}-q}\right)=q^{1/2},
\end{equation*}

leads to the following representation for $\hat{p}\hat{x}$:
\begin{equation*}
    \hat{p}\hat{x}=-iq^{-1/2}\hat{u}\hbar+iq^{1/2}\hat{u}^{-1}\hbar .
\end{equation*}

Similar arguments apply to obtain (\ref{28}). Multiply equation (\ref{26}) by $q$ in order to obtain
\begin{equation}\label{feo}
   q^{-1}\hat{x}\hat{p}-\hat{p}\hat{x} = -iq^{1/2}\hat{u}^{-1}\hbar+iq^{-3/2}\hat{u}^{-1}\hbar,
\end{equation}

adding (\ref{feo}) and (\ref{25}) result
\begin{align*}
    q^{-1}\hat{x}\hat{p}-q\hat{x}\hat{p} & = \ iq^{3/2}\hat{u}\hbar-iq^{1/2}\hat{u}^{-1}\hbar+iq^{-3/2}\hat{u}^{-1}\hbar-iq^{-1/2}\hat{u}\hbar,\\
    \hat{x}\hat{p} & = \ i\left(\frac{q^{3/2}-q^{-1/2}}{q^{-1}-q}\right)\hat{u}\hbar+i\left(\frac{q^{-3/2}-q^{1/2}}{q^{-1}-q}\right)\hat{u}^{-1}\hbar,
\end{align*}

a short calculation going through the steps
\begin{equation*}
  \left(\frac{q^{3/2}-q^{-1/2}}{q^{-1}-q}\right)=-q^{1/2}, \quad  \left(\frac{q^{-3/2}-q^{1/2}}{q^{-1}-q}\right)=q^{-1/2},
\end{equation*}
finally we obtain
\begin{equation*}
    \hat{x}\hat{p}=-iq^{1/2}\hat{u}\hbar+iq^{-1/2}\hat{u}^{-1}\hbar.
\end{equation*}
\end{proof}

\begin{remark}
From (\ref{22}), (\ref{23}), and (\ref{24}) , (\ref{27}) and (\ref{28}) it follows that the set of elements $\left\lbrace\hat{p}^{r}\hat{u}^{n}, \hat{x}^{s}\hat{u}^{n}\vert r, s,n\in\mathbb{N}\right\rbrace$ is a vector space basis of $\mathcal{A}(q)$.
\end{remark}

\subsection{\texorpdfstring{$q$}{Lg}-Heisenberg algebra by Wess - Schwenk}

\begin{definition}[{\cite[p. 1]{Wess-Zumino1990, Chen-Lu-2021}}]
A \textit{\textbf{quantum plane}} is a non commutative algebra generated by $\hat{x}, \hat{p}$ with defining relation 
\begin{equation}\label{quantum-plane}
  \hat{p}\hat{x}=q\hat{x}\hat{p},
\end{equation}

for some fixed $q\neq 0$  
\end{definition}

\begin{definition}[{\cite[p. 1]{Wess-Schwenk1992}}]\label{Wess-Schwenk1992a}
A \textit{\textbf{$q$-Heisenberg algebra over the quantum plane}} is defined by $\hat{x}, \hat{p}$ and $\overline{\hat{x}}$ which is subject to following relations
\begin{align}\label{Heis-Schw}
    \hat{p}\hat{x}-q\hat{x}\hat{p} & = \ -i\hbar, \\
    \label{Heis-Schw1}\hat{p}\overline{\hat{x}}-q^{-1}\overline{\hat{x}}\hat{p} & = \ -iq^{-1}\hbar ,\\ 
    \label{Heis-Schw2}\hat{x}\overline{\hat{x}} & = \ q\overline{\hat{x}}\hat{x}.
\end{align}
\end{definition}
     The $q$-Heisenberg algebra proposed by Schwenk - Wess, defined as a free associative algebra, can be expressed as follows
\begin{equation*}
        \mathcal{H}^{WS}=\frac{\mathbb{C}\left\{ \hat{p}, \hat{x}, \Lambda \right\} }{ \langle \hat{p}\hat{x} - q\hat{x}\hat{p} - i\hbar, \ \hat{p} \overline{\hat{x}} - q^{-1}\overline{\hat{x}} \hat{p} -iq^{-1}\hbar , \ \hat{x} \overline{\hat{x}} - q \overline{\hat{x}} \hat{x} \mid q\in \mathbb{R}, \ q\neq 0 \rangle}.
    \end{equation*}
\begin{proposition}
 The iterated Ore extensions for the skew ring associated to set of relations \rm(\ref{Heis-Schw}), \rm(\ref{Heis-Schw1}) and \rm(\ref{Heis-2}), $R=\mathbb{C}[\hat{x}][\overline{\hat{x}};\sigma_{\overline{\hat{x}}},\delta_{\overline{\hat{x}}}][\hat{p};\sigma_{p},\delta_{p}]$, are given by
 \begin{align}
    \label{s1}\sigma_{\overline{\hat{x}}}(\hat{x}) & = \ q^{-1}\hat{x}, \quad \delta_{\overline{\hat{x}}}(\hat{x})=0,\\ \label{s2}\sigma_{\hat{p}}(\overline{\hat{x}}) &= \ q^{-1}\overline{\hat{x}},\quad  \delta_{\hat{p}}(\overline{\hat{x}})  = -i\hbar q^{-1},\\
    \label{s3}\sigma_{\hat{p}}(\hat{x}) & = \ q\hat{x}, \quad \delta_{\hat{p}}(\hat{x})=-i\hbar .
 \end{align}
\end{proposition}
\begin{proof}
    We can now proceed analogously to the Proof of \ref{Rings-wess} to obtain the following: for the relation(\ref{Heis-Schw2}) 
    \begin{align*}
        \overline{\hat{x}}\hat{x} & = \ \sigma_{\overline{\hat{x}}}(\hat{x})\overline{\hat{x}}+\delta_{\overline{\hat{x}}}(\hat{x}),\\
        \sigma_{\overline{\hat{x}}}(\hat{x}) & = \ q^{-1}\hat{x}, \quad \delta_{\overline{\hat{x}}}(\hat{x})=0.
    \end{align*}
    Now, for (\ref{Heis-Schw1}) we have
    \begin{align*}
        \hat{p}\overline{\hat{x}} & = \ \sigma_{\hat{p}}(\overline{\hat{x}})\hat{p}+\delta_{\hat{p}}(\overline{\hat{x}}),\\
        \sigma_{\hat{p}}(\overline{\hat{x}}) & = \ q^{-1}\overline{x}, \quad \delta_{\hat{p}}(\overline{\hat{x}})=-i\hbar q^{-1},
    \end{align*}
    and finally for (\ref{Heis-Schw}) 
    \begin{align*}
        \hat{p}\hat{x} & = \ \sigma_{\hat{p}}(\hat{x})+\delta_{\hat{p}}(\hat{x}),\\
        \sigma_{\hat{p}}(\hat{x}) & = \ q\hat{x}, \quad \delta_{\hat{p}}(\hat{x}) = -i\hbar .
    \end{align*}
\end{proof}
\subsection{\texorpdfstring{$p,q-$}{Lg}Heisenberg enveloping algebra}\label{env}

\begin{definition}[{\cite[p. xix]{Goodearl-Warfield}}]
    A \textit{\textbf{Lie algebra}} is a $\Bbbk$-vector space $\textit{\textbf{g}}$ equipped with a lie product, i.e a $\Bbbk$-bilinear map 
    \begin{align*}
    \textit{\textbf{g}}\times\textit{\textbf{g}}\longrightarrow \textit{\textbf{g}}: (\hat{x}, \hat{y})\longrightarrow [\hat{x}, \hat{y}],
    \end{align*}
    satisfying 
    \begin{align*}
    [\hat{x}, \hat{x}] & = \ [\hat{y}, \hat{y}]=0, \quad [\hat{x}, \hat{y}]=-[\hat{y}, \hat{x}].
    \end{align*}
    and the \textit{\textbf{Jacobi identity}} 
    \begin{equation*}
    [\hat{x},[\hat{y}, \hat{z}]]+[\hat{y},[\hat{z}, \hat{x}]]+[\hat{z},[\hat{x}, \hat{y}]]=0,
    \end{equation*}
    for all $\hat{x},\hat{y}, \hat{z}\in\textit{\textbf{g}}$.
\end{definition}

\begin{definition}\cite{Cooperstein2015}
Let $X$ be a $\Bbbk$-vector space. An $\Bbbk$-\textit{\textbf{associative algebra}} is a $\Bbbk$-vector space $X$ with a bilinear map $X\times X\longrightarrow X$ denoted by 
\begin{equation*}
\varphi (\hat{x},\hat{y})=\hat{x}\hat{y},
\end{equation*}
which is associative in the sense that 
\begin{equation*}
\hat{x}(\hat{y}\hat{z})=(\hat{x}\hat{y})\hat{z}
\end{equation*}
for all $\hat{x}, \hat{y}, \hat{z}\in X$ and satisfying the following properties
  \begin{itemize}
      \item [\rm (1)] The right and left distributive property holds: for $\hat{x}, \hat{y}, \hat{z}\in X$, 
      \begin{align*}
      \hat{x}(\hat{y}\hat{z}) & = \ (\hat{x}\hat{y})\hat{z},\quad    
      (\hat{x}+\hat{y})\hat{z}=\hat{x}\hat{z}+\hat{y}\hat{z},\quad
      \hat{z}(\hat{x}+\hat{y})=\hat{z}\hat{x}+\hat{z}\hat{y}.
      \end{align*}
      \item [\rm (2)] For all $\alpha ,\hat{x}, \hat{y}\in X$ 
      \begin{equation*}
      (\alpha\hat{x})\hat{y}=\hat{x}(\alpha\hat{y})=\alpha (\hat{x}\hat{y}).
      \end{equation*}
  \end{itemize}
\end{definition}
\begin{definition}
    An algebra $X$ is unital if there exists an element $\mathbf{1}\in X$ such that $\mathbf{1}\hat{x}=\hat{x}\mathbf{1}=\hat{x}$ for all $\hat{x}\in X$. In this definition associative algebras are assumed to have a multiplicative identity, denoted $\mathbf{1}$; they are sometimes called \textit{\textbf{unital associative algebras}}. 
\end{definition}

\begin{definition}\cite{Hall2010}
    Let $\textit{\textbf{g}}$ and $\textit{\textbf{h}}$ be Lie algebras. A linear map $\phi : \textit{\textbf{g}}\longrightarrow\textit{\textbf{h}}$ is called \textit{\textbf{Lie algebra Homomorphism}} if 
    \begin{equation*}
    \phi ([\hat{x}, \hat{y}])=[\phi (\hat{x}),\phi (\hat{y})],
    \end{equation*}
    for all $\hat{x}, \hat{y}\in\textit{\textbf{g}}$.
\end{definition}

\begin{definition}\cite{kac2010}
   If \(\textit{\textbf{g}}\) is a Lie algebra, then the \(\Bbbk\)-\(\textit{\textbf{enveloping algebra}}\) of \(\textit{\textbf{g}}\) is a pair \((\phi, \mathcal{U})\), where \(\mathcal{U}\) is a unital associative algebra and \(\phi: \textit{\textbf{g}} \longrightarrow \mathcal{U}\) is a homomorphism of Lie algebras, with \(\mathcal{U}\) equipped with the bracket
\[
[\hat{x}, \hat{y}] = \hat{x}\hat{y} - \hat{y}\hat{x}.
\]
\end{definition}

With this in mind, we have the following definition about the $q$-Heisenberg algebra proposed by Gaddis \cite{{Jasson-Gaddis2016}}.
\begin{definition}\cite{Jasson-Gaddis2016}
    If $\hat{x}, \hat{y}$ and $\hat{z}$ are the generators of the Lie algebra then the \textit{\textbf{Heisenberg Lie algebra}} $\textit{\textbf{h}}$  is defined on generators $\hat{x},\hat{y},\hat{z}$ subject to the relations
    \begin{equation*}
    [\hat{x}, \hat{y}]=\hat{z}, \quad [\hat{x}, \hat{z}]=[\hat{y},\hat{z}]=0.
    \end{equation*}
\end{definition}

\begin{definition}\cite{Jasson-Gaddis2016}
   Let $\hat{x},\hat{y},\hat{z}$ be generators.  For $p,q\in\mathbb{R}^{+}$, the \textit{\textbf{quantum Heisenberg enveloping algebra}} denoted by $\mathcal{H}_{p,q}$ is the $\Bbbk$-algebra generated by $\hat{x},\hat{y},\hat{z}$ subject to the relations
\begin{align}\label{29a}
    \hat{z}\hat{x} & = \ q^{-1}\hat{x}\hat{z},\\
    \label{env-Heis1}\hat{z}\hat{y} & = \ p\hat{y}\hat{z}, \\
    \label{env-Heis2}\hat{y}\hat{x}-q\hat{x}\hat{y} & = \ \hbar\hat{z}.
\end{align} 
\end{definition}
Notice that when $p=q$, the algebra $\mathcal{H}_{p,q}$ becomes $\mathcal{H}_{q}$ and hence $\mathcal{H}_{p,q}$ generalizes the one-parameter analog $\mathcal{H}_{q}$. On other hand, the same author defines the $[k]_{p,q}-$ number as
\begin{equation}\label{29b}
    [k]_{p,q} := \frac{q^{k}-p^{-k}}{q-p^{-1}}=\sum\limits_{i=0}^{k-1}q^{i}p^{-(k-1-i)},
\end{equation}
and in $\mathcal{H}_{p,q}$ the following identities hold for $k\geq 0$:
\begin{align}
    \label{29c}\hat{y}\hat{x}^{k} & = \ q^{k}\hat{x}^{k}\hat{y}+\hbar [k]_{p,q}\hat{x}^{k-1}\hat{z},\\
    \label{29d}\hat{y}^{k}\hat{x} & = \ q^{k}\hat{x}\hat{y}^{k}+\hbar[k]_{p,q}\hat{z}\hat{y}^{k-1}.
\end{align}
If $\hat{x}, \hat{p}$ satisfies the relation (\ref{29a}) then the quantum Heisenberg enveloping algebra can be written as 
    \begin{equation*}
        \mathcal{H}_{p.q}=\frac{\Bbbk\left\{ \hat{x}, \hat{y}, \hat{z} \right\} }{ \langle \hat{z}\hat{x} - q^{-1} \hat{x}\hat{z}, \ \hat{z}\hat{y} - p\hat{y} \hat{z}, \ \hat{y} \hat{x} - q \hat{x} \hat{y}-\hbar\hat{z} \mid p,q\in \Bbbk, \ p,q\neq 0 \rangle}.
    \end{equation*}

\begin{proposition}
  The skew polynomials associated to \rm(\ref{29a}), $R=\Bbbk [\hat{x}][\hat{z};\sigma_{\hat{z}},\delta_{\hat{z}}][\hat{y};\sigma_{\hat{y}}, \delta_{\hat{y}}]$ are given by
  \begin{align}
     \sigma_{\hat{y}}(\hat{x}) & = \ q\hat{x}, \quad \delta_{\hat{y}}(\hat{x})=\hbar\hat{z},\\
      \sigma_{\hat{z}}(\hat{y}) & = \ p\hat{y}, \quad \delta_{\hat{z}}(\hat{y})=0,\\
      \sigma_{\hat{z}}(\hat{x}) & = \ q^{-1}\hat{x}, \quad \delta_{\hat{z}}(\hat{x})=0.
  \end{align}
  \end{proposition}
\begin{proof}
    For (\ref{29a}) we have that
    \begin{align*}
    \hat{z}\hat{x} & = \ \sigma_{\hat{z}}(\hat{x})\hat{z}+\delta_{\hat{z}}(\hat{x}),\\
    \sigma_{\hat{z}}(\hat{x}) & =\ q^{-1}\hat{x}, \quad \delta_{\hat{z}}(\hat{x})=0.
    \end{align*}
    Now for (\ref{env-Heis1}) 
    \begin{align*}
    \hat{z}\hat{y} & = \sigma_{\hat{z}}(\hat{y})\hat{z}+\delta_{\hat{z}}(\hat{y}),\\
    \sigma_{\hat{z}}(\hat{y}) & = \ p\hat{y}, \quad \delta_{\hat{z}}(\hat{y})=0,
    \end{align*}
    and finally for (\ref{env-Heis2})
    \begin{align*}
        \hat{y}\hat{x} & = \ \sigma_{\hat{y}}(\hat{x})\hat{y}+\delta_{\hat{y}}(\hat{x}),\\
        \sigma_{\hat{y}}(\hat{x}) & = \ q\hat{x}, \quad \delta_{\hat{y}}(\hat{x})=\hbar\hat{z}.
    \end{align*}
\end{proof}

\subsection{Generalized Heisenberg algebra}

\begin{definition}\cite{OuyangLiuXiang2015}
If $f(\hat{h})\in\mathbb{C}[\hat{h}]$ is a fixed polynomial over $\mathbb{C}$ then the \textit{\textbf{generalized Heisenberg algebra}} $\mathcal{H}(f)$ is the associative algebra over $\mathbb{C}$  generated by $\hat{x}, \hat{y}, \hat{h}$ with definition relations
\begin{equation}\label{30}
    \hat{h}\hat{x} = \hat{x}f(\hat{h}), \quad \hat{y}\hat{h} = f(\hat{h})\hat{y}, \quad \hat{y}\hat{x}-\hat{x}\hat{y} = \hbar f(\hat{h})-\hbar\hat{h}.
\end{equation}
    According to relations (\ref{30}) the generalized Heisenberg algebra can be written of the following form
    \begin{equation*}
        \mathcal{H}(f)=\frac{\mathbb{C}\left\{ \hat{x}, \hat{y}, \hat{h} \right\}}{\langle \hat{h}\hat{x} - \hat{x}f(\hat{h}), \ \hat{y}\hat{h} - f(\hat{h}) \hat{y}, \ \hat{y} \hat{x} - \hat{x} \hat{y}-\hbar f(\hat{h})+\hbar\hat{h}\rangle}.
    \end{equation*}
\end{definition}
\begin{definition}\cite{Razavinia-Lopes2020} The $q$-\textit{\textbf{generalized Heisenberg algebra}} is the unital associative algebra generated by $\hat{x}, \hat{y}$ and $\hat{h}$ subject to the defining relations
\begin{align}\label{31lr}
 \hat{h}\hat{x} & = \ \hat{x}f(\hat{h}), \\
 \label{32lr}\hat{y}\hat{h}  & = \ f(\hat{h})\hat{y},\\ 
 \label{33lr}\hat{y}\hat{x}-q\hat{x}\hat{y} &  = \ \hbar g(\hat{h}).
\end{align}
    The $q$-generalized Heisenberg algebra can be expressed as \begin{equation*}
        \mathcal{H}_{q}(f,g):=\frac{\Bbbk\left\{ \hat{x}, \hat{y}, \hat{h} \right\}}{\langle \hat{h}\hat{x} - \hat{x}f(\hat{h}), \ \hat{y}\hat{h} - f(\hat{h}) \hat{y}, \ \hat{y} \hat{x} - q\hat{x} \hat{y}-\hbar g(\hat{h})\rangle}.
    \end{equation*}
\end{definition}
\begin{example}\cite{Razavinia-Lopes2020}
 The {\em \textbf{quantum generalized Heisenberg}} is defined by the commutation relations $\hat{h}\hat{x}=\hat{x}f(\hat{h}), \hat{y}\hat{h}=f(\hat{h})\hat{y}$ and $\hat{y}\hat{x}-\hat{x}\hat{y}=\ \hbar f(\hat{h})-\hbar\hat{h}$ (non-deformed, in the case deformed $\hat{y}\hat{x}=q\hat{x}\hat{y}+\hbar g(\hat{h})$).
\end{example}

\subsection{\texorpdfstring{$q-\hbar$}{Lg} Heisenberg algebra}
\begin{definition}\cite{Arefeva-Volovich1991}
    Let $\hat{x}_{j}, \hat{p}_{k}$ be generators of position and momentum that satisfy the relation (\ref{quantum-plane}). The \textit{\textbf{Quantum space phase}} denoted by $\mathbb{R}^{2}_{q}$ is the same quantum plane, such that $\hat{x}_{j}$ is the coordinate on the \textit{\textbf{quantum line}}.
\end{definition}

\begin{definition}\cite{Arefeva-Volovich1991}\label{Arefeva-Volovich1991definition2}
    The $q-\hbar$ \textit{\textbf{Heisenberg algebra}} is the algebra generated by $\hat{x}_{j}$ and $\hat{p}_{k}$ satisfying the relation
    \begin{equation}\label{h-bar-Heis}
        \hat{p}_{k}\hat{x}_{j}-q\hat{x}_{j}\hat{p}_{k}=-i\hbar q^{1/2},
    \end{equation}
    
where $q\in\mathbb{C}-\left\lbrace 0\right\rbrace$.
\end{definition}

\begin{definition}\cite{Arefeva-Volovich1991} \label{Arefeva-Volovich1991definition}
     If $\hat{x}_{i},\hat{p}_{i}$ are generators over the quantum space phase, then $q-\hbar$ \textit{\textbf{Heisenberg algebra quantization}} is the algebra defined by the relation 
\begin{equation}\label{quantization}
     \hat{x}_{j}\hat{p}_{k}-q\hat{p}_{k}\hat{x}_{j}=i\hbar D_{jk}(q),
\end{equation}
     
where $D_{jk}(q)$ any function that depends on $q$.
\end{definition}

\begin{example}\cite{Arefeva-Volovich1991}
    The $q-\hbar$ \textit{\textbf{generalized Heisenberg algebra}} can be expressed as
    \begin{equation*}
        \mathcal{Q}_{q}(\hat{x}_{j}, \hat{p}_{k})=\frac{\mathbb{R}^{2}_{q}\left\{ \hat{x}_{j}, \hat{p}_{k} \right\}}{\langle \hat{p}_{k}\hat{x}_{j} - q\hat{x}_{j}\hat{p}_{k}+i\hbar q^{1/2}, \ \hat{x}_{j}\hat{p}_{k} - q\hat{p}_{k}\hat{x}_{j}-i\hbar D_{jk}(q), \ \mid q\in\mathbb{C}, j,k\in\mathbb{N}\rangle}.
    \end{equation*}
\end{example}

\section{A new \texorpdfstring{$q$-$\hbar$}{Lg} Heisenberg algebra}\label{propose-Heis}
Next, we propose a deformation of the Heisenberg algebra, which generalizes certain previously studied deformed algebras.

\begin{definition}\label{NPh}
A \textbf{new $q$-$\hbar$ Heisenberg algebra} $\boldsymbol{\mathcal{H}}_{q}$ is the $\Bbbk$-algebra generated by the set of indeterminates $\hat{x}_{\alpha}, \hat{y}_{\lambda}, \hat{p}_{\beta}$ with $\alpha, \beta, \lambda \in \{1, 2, 3\}$ subject to the relations
\begin{align}
    \label{nH1}
    \hat{x}_{\alpha}\hat{p}_{\beta} - q^{n}\hat{p}_{\beta}\hat{x}_{\alpha} &= iq^{n-1}\hbar^{n}\Psi(q; \hat{x}_{\alpha}, \hat{y}_{\lambda},\hat{p}_{\beta}), \\
    \label{nH2}q^{m}\hat{x}_{\alpha}\hat{y}_{\lambda} - \hat{y}_{\lambda}\hat{x}_{\alpha} &= -i(q-1)^{m-1}\hbar^{m-1}\Pi (q;\hat{x}_{\alpha},\hat{y}_{\lambda},\hat{p}_{\beta}), \\
    \label{nH3}q^{l}\hat{y}_{\lambda}\hat{p}_{\beta} - q^{l+1}\hat{p}_{\beta}\hat{y}_{\lambda} &= i\hbar^{l}\Phi (q; \hat{x}_{\alpha},\hat{y}_{\lambda}, \hat{p}_{\beta}), \quad n, l, m \in \mathbb{R},
\end{align}
where $\Psi (q; \hat{x}_{\alpha}, \hat{y}_{\lambda}, \hat{p}_{\beta})$ and $\Phi (q; \hat{x}_{\alpha},\hat{y}_{\lambda}, \hat{p}_{\beta})$ are arbitrary functions depending on $q$, $\hat{x}_{\alpha}$, $\hat{y}_{\lambda}$ and $\hat{p}_{\beta}$ called \textit{\textbf{dynamical functions}}, for $q \in\mathbb{R}-\left\lbrace 0,1\right\rbrace$.
\end{definition}
A new $q$-$\hbar$ Heisenberg algebra also can be written in the following form.
    \begin{tiny}
    \begin{multline*}
      \boldsymbol{\mathcal{H}}_{q}
      =\\\frac{\Bbbk  \left\lbrace\hat{x}_{\alpha},\hat{y}_{\lambda}, \hat{p}_{\lambda}\right\rbrace}{\langle\hat{x}_{\alpha}\hat{p}_{\beta}-q^{n}\hat{p}_{\beta}\hat{x}_{\alpha}-  iq^{n-1}\hbar^{n}\Psi (q;\hat{x}_{\alpha},\hat{y}_{\lambda},\hat{p}_{\beta}),q^{m}\hat{x}_{\alpha}\hat{y}_{\lambda}-\hat{y}_{\lambda}\hat{x}_{\alpha}+iq^{m-1}\hbar^{m-1}, q^{l}\hat{y}_{\lambda}\hat{p}_{\beta}-q^{l-1}\hat{p}_{\beta}\hat{y}_{\lambda}-i\hbar^{l}\Phi (q;\hat{x}_{\alpha},\hat{y}_{\lambda},\hat{p}_{\beta})\vert q\in\mathbb{R}-\left\lbrace 0, 1\right\rbrace\rangle}. 
    \end{multline*}
    \end{tiny}

\begin{remark}
    For $n=1, l=1, m=1$ and $q=1$, we obtain the "classical" Heisenberg algebra (\ref{Heis-1}) and (\ref{Heis-2}) when $\Psi =1$ and $\Pi = \Phi =0$.
\end{remark}
\begin{remark}
    The expression \eqref{17} can be written as
    \begin{equation*}
        \hat{x}\hat{p}-q^{-1}\hat{p}\hat{x}=i\hbar\hat{\Lambda}q^{-1/2}.
    \end{equation*}
\end{remark}
\begin{proof}
    \begin{align*}
        q^{1/2}\hat{x}\hat{p}-q^{-1/2}\hat{p}\hat{x} & = \ i\hbar\hat{\Lambda},\\
        \frac{1}{q^{-1/2}}\hat{x}\hat{p}-q^{-1/2}\hat{p}\hat{x} & = \ i\hbar\hat{\Lambda},\\
      \hat{x}\hat{p}-q^{-1}\hat{p}\hat{x} & = \ i\hbar\hat{\Lambda}q^{-1/2}.  
    \end{align*}
\end{proof}
Taking into account the above, we can consider the following example
\begin{example}
For the specific case where \( n = -1 \), \( l = -1 \), and \( m = -1 \), the dynamical functions are given by:
\begin{align*}
\Psi &= \hbar^{2}\hat{\Lambda}q^{3/2}, \\
\Phi &= 0, \\
\Pi &= 0.
\end{align*}
These solutions satisfy the $q$-deformed Heisenberg algebra defined by relations \eqref{17}, \eqref{wess2}, and \eqref{wess3}. Starting from expression \eqref{nH1} for $\alpha = \beta = \lambda = 1$, with the substitutions \( \hat{x}_{1}=\hat{x}\), \( \hat{p}_{1}= \hat{p} \), and \( \hat{y}_{1}=\hat{\Lambda} \), we obtain for \( n = -1 \):
\begin{equation}\label{reduced}
\hat{x}\hat{p} - q^{-1}\hat{p}\hat{x} = iq^{-2}\hbar^{-1}\Psi.
\end{equation}

Equating \eqref{reduced} with \eqref{17}:
\begin{equation*}
i\hbar\hat{\Lambda}q^{-1/2} = iq^{-2}\hbar^{-1}\Psi,
\end{equation*}
we directly solve for $\Psi$:
\begin{equation*}
\Psi = \hbar^{2}\hat{\Lambda}q^{3/2}.
\end{equation*}
For \( m = -1 \), relation \eqref{wess2} becomes:
\begin{equation*}
q^{-1}\hat{x}\hat{\Lambda} - \hat{\Lambda}\hat{x} = -i(q-1)^{-2}\hbar^{-2}\Pi.
\end{equation*}
Comparison with \eqref{nH2} shows that both sides vanish identically, thus:
\begin{equation*}
\Pi = 0.
\end{equation*}
From \eqref{nH3} with \( l = 0 \), we have:
\begin{equation*}
\hat{\Lambda}\hat{p} - q\hat{p}\hat{\Lambda} = i\Phi.
\end{equation*}
The left-hand side vanishes by algebraic consistency, requiring:
\begin{equation*}
\Phi = 0.
\end{equation*}

This completes the verification that all dynamical functions satisfy the $q$-deformed algebra for the specified parameters.
\end{example}
\begin{example}
We derive the dynamical functions $\Psi$, $\Phi$, and $\Pi$ for different parameter values, showing consistency with the algebraic relations \eqref{rem2}. For $\alpha = \beta = \lambda = 1$, with the substitutions \( \hat{x}_{1}=\hat{x}\), \( \hat{p}_{1}= \hat{p} \), and \( \hat{y}_{1}=\hat{u} \), and for the specified parameter values, we obtain:
\begin{align}
\text{For } n = 1:& \quad \Psi = (q^{-1/2} - q^{3/2})\hat{u}^{-1}, \label{psi1} \\
\text{For } n = -1:& \quad \Psi = (q^{1/2} - q^{5/2})\hbar^{2}\hat{u} \label{psi2}, \\
\text{For } l = 0:& \quad \Phi = 0 \label{phi0}, \\
\text{For } m = 1:& \quad \Pi = 0 \label{pi0}.
\end{align}
Beginning with relation \eqref{nH3} and setting $l = 0$, we obtain:
\begin{equation*}
\hat{u}\hat{p} - q\hat{p}\hat{u} = i\Phi .\label{eq:phi}
\end{equation*}
Comparison with \eqref{22} shows both sides must vanish, yielding:
\begin{equation}
\Phi = 0.
\end{equation}
From \eqref{nH2} with $m = -1$, we have:
\begin{equation*}
q^{-1}\hat{x}\hat{u} - \hat{u}\hat{x} = -i\hbar^{-2}(q-1)^{-2}\Pi ,\label{eq:pi}
\end{equation*}
matching with \eqref{23} requires:
\begin{equation*}
\Pi = 0.
\end{equation*}
Substituting $n = 1$ into \eqref{nH1} gives:
\begin{equation}
\hat{x}\hat{p} - q\hat{p}\hat{x} = i\hbar\Psi .\label{eq:psi1}
\end{equation}
Equating with \eqref{25}:
\begin{align*}
-i\hbar (q^{3/2} - q^{-1/2})\hat{u}^{-1} &= i\hbar\Psi , \\
\Psi &= (q^{-1/2} - q^{3/2})\hat{u}^{-1}.
\end{align*}
For $n = -1$ in \eqref{nH1}, we obtain:
\begin{equation*}
\hat{x}\hat{p} - q^{-1}\hat{p}\hat{x} = i\hbar^{-1}q^{-2}\Psi .\label{eq:psi2}
\end{equation*}
Rewriting and comparing with \eqref{26}:
\begin{align*}
q\hat{x}\hat{p} - \hat{p}\hat{x} &= -i\hbar^{-1}q^{-1}\Psi  ,\\
-i\hbar^{-1}q^{-1}\Psi &= i\hbar\hat{u}(q^{3/2} - q^{-1/2}), \\
\Psi &= (q^{1/2} - q^{5/2})\hbar^{2}\hat{u}.
\end{align*}

All derived dynamical functions \eqref{psi1}-\eqref{pi0} satisfy the $q$-deformed algebra relations for their respective parameter values, demonstrating consistency across the different cases considered.
\end{example}

\begin{example}
We demonstrate that the dynamical functions
\begin{align}
\Psi &= q\hbar^{2}, \label{Psi-result} \\
\Phi &= q^{-1}\hbar^{2}, \label{Phi-result} \\
\Pi &= 0, \label{Pi-result}
\end{align}
satisfy the relations \eqref{Heis-Schw}, \eqref{Heis-Schw1}, and \eqref{Heis-Schw2} for the parameter values \( n = -1 \), \( m = -1 \), and \( l = 0 \). The derivations employ the substitutions \( \hat{x}_{1} = \hat{x} \), \( \hat{p}_{1} = \hat{p} \), and \( \hat{y}_{1} = \overline{\hat{x}} \).

Beginning with relation \eqref{nH1} for \( n = -1 \), we obtain:
\begin{align}
\hat{x}\hat{p} - q^{-1}\hat{p}\hat{x} &= iq^{-2}\hbar^{-1}\Psi , \label{Psi-start} \\
\Rightarrow q\hat{x}\hat{p} - \hat{p}\hat{x} &= iq^{-1}\hbar^{-1}\Psi , \label{Psi-intermediate} \\
\Rightarrow \hat{p}\hat{x} - q\hat{x}\hat{p} &= -iq^{-1}\hbar^{-1}\Psi . \label{Psi-final-form}
\end{align}

Equating \eqref{Psi-final-form} with \eqref{Heis-Schw} yields:
\begin{align*}
-i\hbar = -i\hbar^{-1}q^{-1}\Psi  
\Rightarrow \Psi &= q\hbar^{2} ,
\end{align*}
This establishes result \eqref{Psi-result}. For relation \eqref{Heis-Schw1}, we consider \eqref{nH2} with \( l = -1 \):
\begin{align*}
q^{-1}\overline{\hat{x}}\hat{p} - \hat{p}\overline{\hat{x}} &= i\hbar^{-1}\Phi 
\Rightarrow \hat{p}\overline{\hat{x}} - q^{-1}\overline{\hat{x}}\hat{p} = -i\hbar^{-1}\Phi .
\end{align*}
Comparison with \eqref{nH2} gives:
\begin{align*}
-iq^{-1}\hbar = -i\hbar^{-1}\Phi 
\Rightarrow \Phi = q^{-1}\hbar^{2} ,
\end{align*}
This confirms result \eqref{Phi-result}.For relation \eqref{Heis-Schw2}, we examine \eqref{nH3} with \( m = -1 \):
\begin{align*}
q^{-1}\hat{x}\overline{\hat{x}} - \overline{\hat{x}}\hat{x} &= -i\hbar^{-2}(q-1)^{-2}\Pi , \label{Pi-start} \\
\Rightarrow \hat{x}\overline{\hat{x}} - q\overline{\hat{x}}\hat{x} &= -i(q-1)^{-2}q\hbar^{-2}\Pi , 
\end{align*}

the algebraic consistency requires:
\begin{equation*}
\Pi = 0 . 
\end{equation*}
This verifies result \eqref{Pi-result} and recovers \eqref{Heis-Schw2}.The derived dynamical functions \eqref{Psi-result}-\eqref{Pi-result} satisfy all required algebraic relations for the specified parameter values, demonstrating the consistency of the $q$-deformed Heisenberg algebra in this case.
\end{example}
\begin{example}
We examine the dynamical function $\Psi$ within the context of the $q$-$\hbar$ algebra, establishing its explicit form through two distinct but related approaches. The function $\Psi$ admits two equivalent characterizations:
From relation \eqref{h-bar-Heis}:
\begin{equation}
\Psi  = \hbar^{2}q^{3/2}, \label{Psi-hbar} 
\end{equation}
and from relation \eqref{quantization}:
\begin{equation}
\Psi  = D_{jk}(q), \label{Psi-quant}
\end{equation}
valid for all deformation parameters $q \in \mathbb{C} \setminus \{0, 1\}$. Consider the operator substitutions $\hat{x}_{1} = \hat{x}_{j}$ and $\hat{p}_{1} = \hat{p}_{k}$. For the case $n = -1$, relation \eqref{nH1} transforms as follows:
\begin{align}
\hat{x}_{j}\hat{p}_{k} - q^{-1}\hat{p}_{k}\hat{x}_{j} &= i\hbar^{-1}q^{-2}\Psi, \label{start-deriv} \\
q\hat{x}_{j}\hat{p}_{k} - \hat{p}_{k}\hat{x}_{j} &= iq^{-1}\hbar^{-1}\Psi, \label{interm-deriv} \\
\hat{p}_{k}\hat{x}_{j} - q\hat{x}_{j}\hat{p}_{k} &= -iq^{-1}\hbar^{-1}\Psi. \label{final-deriv}
\end{align}

Comparing \eqref{final-deriv} with the reference relation \eqref{h-bar-Heis} yields the identity:
\begin{align}
-i\hbar q^{1/2} &= -i\hbar^{-1}q^{-1}\Psi, \label{identity} \\
\hbar^{2}q^{3/2} &= \Psi, \label{solution}
\end{align}
establishing the first characterization \eqref{Psi-hbar}. For the alternative case $n = 1$, relation \eqref{nH1} takes the form:
\begin{equation*}
\hat{x}_{j}\hat{p}_{k} - q\hat{p}_{k}\hat{x}_{j} = i\hbar\Psi. 
\end{equation*}

Direct comparison with the quantization condition \eqref{quantization} produces:
\begin{align*}
i\hbar D_{jk}(q) & = i\hbar\Psi,  \\
D_{jk}(q) & = \Psi, 
\end{align*}
confirming the second characterization \eqref{Psi-quant}. The dual representation of $\Psi$ reveals:
\begin{itemize}
\item Its fundamental role as $\hbar^{2}q^{3/2}$ in the algebraic structure
\item Its interpretation as the deformation function $D_{jk}(q)$ in quantization schemes
\end{itemize}
This consistency underscores the robustness of the $q$-deformed Heisenberg algebra framework. The dynamical function $\Psi$ exhibits remarkable versatility, serving both as a structural component of the $q$-$\hbar$ algebra and as a quantization parameter, with its form adapting naturally to each mathematical context.
\end{example}

The following table shows that for certain values of $n, m, l , q, \Psi$, and $\Phi$, some relations of the previously studied deformed algebras are obtained.

\begin{landscape}  
\begin{table}[ht]
\centering
\begin{tabular}{|c|c|c|c|c|c|c|c|} 
  \hline
    \bf{Some Heisenberg algebras}  & $n$ & $l$ & $m$ & $q$ & $\Psi$ & $\Pi$ & $\Phi$ \\ \hline
  Classic (Definition \ref{classic-Heis}) & 1  & 1 & 1 & 1 & 1 & 0 & 1\\ \hline
  Classic (Definition \ref{classic-Heis}) & 0 & 0 & 0 & $\neq 1$ & 0 & 0 & 0\\ \hline
  Wess (Definition \ref{Wess-def}) & -1 & 0 & 0  &  $\mathbb{R}$ & 0 & $\hbar^{2}\hat{\Lambda}q^{3/2}$ & 0 \\ \hline
  Wess (Definition \ref{Wess-def}) & 0 & 0 & -1 & $\mathbb{R}$ & 0 & 0 & 0\\ \hline
  Wess (Definition \ref{Wess-def}) & 0 & -1 & 0 & $\mathbb{R}$ & 0 & 0 & 0\\ \hline
  Schmudgen (Definition \ref{schmudgena})  & 0 & -1 & 0 & $\neq -1$ & 0 & 0 & 0\\ \hline
  Schmudgen (Definition \ref{schmudgena}) & 0 & 0 & -1 & $\neq -1$ & 0 & 0 & 0\\ \hline
  Schmudgen (Definition \ref{schmudgena}) & -1 & 0 & 0 & $\neq -1$ & $\hbar^{2}\hat{u}(q^{1/2}-q^{5/2})$ & 0 & 0\\ \hline
  Schmudgen (Definition \ref{schmudgena}) & 1 & 0 & 0 & $\neq -1$ & $(q^{3/2}-q^{-1/2})\hat{u}^{-1}$ & 0 & 0 \\ \hline
  Schwenk-Wess (Definition \ref{Wess-Schwenk1992a}) & -1 & 0 & 0 & $\neq 0$ & $q\hbar^{2}$ & 0 & 0\\ \hline
  Schwenk-Wess (Definition \ref{Wess-Schwenk1992a}) & 0 & -1 & 0 & $\neq 0$ & 0 & 0 & $q^{-1}\hbar^{2}$\\ \hline
   Schwenk-Wess (Definition \ref{Wess-Schwenk1992a}) & 0 & 0 & -1 & $\neq 0$ & 0 & 0 & 0\\ \hline
  $q-\hbar$ (Definition \ref{Arefeva-Volovich1991definition2}) & -1 & 0 & 0 & $\mathbb{C}-\left\lbrace 0,1\right\rbrace$ & $\hbar^{2}q^{3/2}$ & 0 & 0\\ \hline
  $q- \hbar$ quantization (Definition \ref{Arefeva-Volovich1991definition}) & -1 & 0 & 0 & $\mathbb{C}-\left\lbrace 0,1\right\rbrace$ & $D_{jk}(q)$ & 0 & 0\\ \hline
  
\end{tabular}
\caption{Values for $n, m, l, q, \Psi$, and $\Phi$ for some deformations of Heisenberg algebra}
\label{tab:heisenberg_algebras}
\end{table}
\end{landscape}

\bibliographystyle{plain} 
\bibliography{Bibliography}

\end{document}